\documentclass[a4paper, 10pt, onecolumn, conference]{ieeeconf}  

\IEEEoverridecommandlockouts                              
\usepackage{hyperref}

\usepackage{graphicx}

\usepackage{amsmath}
\usepackage{amssymb}
\usepackage{amsfonts}
\usepackage{times}
\usepackage{mathptmx} 
\usepackage{datetime}
\newtheorem{lem}{Lemma}
\newtheorem{thm}{Theorem}

\newcommand{\tanc}{\mathrm{tanc}}
\newcommand{\tanhc}{\mathrm{tanhc}}

\DeclareMathAlphabet{\bit}{OML}{cmm}{b}{it}
\def\<{\leqslant}           
\def\>{\geqslant}           

\def\d{\partial}

\def\Re{\mathrm{Re}}   
\def\Im{\mathrm{Im}}   

\def\mN{\mathbb{N}}    
\def\mR{{\mathbb R}}    
\def\mC{\mathbb{C}}    

\def\Tr{\mathrm{Tr}}       
\def\rT{{\rm T}}        
\def\rF{\mathrm{F}}        
\def\rHS{\mathrm{HS}}        

\def\bE{\mathbf{E}}    
\def\bM{\mathbf{M}}    

\def\bra{{\langle}}
\def\ket{{\rangle}}

\def\re{{\rm e}}        
\def\rd{{\rm d}}        

\def\cL{\mathcal{L}}

\def\bJ{\mathbf{J}}

\def\br{\mathbf{r}}
\def\x{\times}
\def\ox{\otimes}

\def\fF{\mathfrak{F}}

\def\fH{\mathfrak{H}}

\def\cK{\mathcal{K}}

\def\cI{\mathcal{I}}
\def\cP{\mathcal{P}}

\title{\LARGE \bf
A Girsanov Type Representation of Quadratic-Exponential Cost Functionals 
for Linear Quantum Stochastic  Systems$^*$}

\author{Igor G. Vladimirov$^\dagger$,\qquad
Ian R. Petersen$^\dagger$, \qquad
Matthew R. James$^\dagger$%
\thanks{$^*$This work is supported by the Air Force Office of Scientific Research (AFOSR) under agreement number FA2386-16-1-4065 and the Australian Research Council under grant DP180101805.}
\thanks{$^\dagger$Research School of Electrical, Energy and Materials Engineering, College of Engineering and Computer Science,
Australian National University, Canberra, Acton, ACT 2601,
Australia, {\tt\scriptsize
igor.g.vladimirov@gmail.com, i.r.petersen@gmail.com, matthew.james@anu.edu.au.
}
}
}

\begin{document}

\maketitle
\thispagestyle{empty}
\pagestyle{plain}

\begin{abstract}
This paper is concerned with multimode open quantum harmonic oscillators and quadratic-ex\-po\-nen\-ti\-al functionals (QEFs) as quantum risk-sensitive performance criteria. Such systems are described by linear quantum stochastic differential equations driven by multichannel bosonic fields.
We develop a finite-horizon expansion for the system variables using the eigenbasis of their two-point commutator kernel with 
noncommuting position-momentum pairs as coefficients.  This quantum Karhunen-Loeve expansion is 
used in order to obtain a Girsanov type representation for the quadratic-exponential functions of the system variables. This representation is valid regardless of a particular system-field state and employs the averaging  over an auxiliary  classical Gaussian random process whose covariance operator is defined in terms of the quantum commutator kernel. We use this representation in order to relate the QEF to the moment-generating functional of the system variables.  This result is also specified for the invariant multipoint Gaussian quantum state when the oscillator is driven by vacuum fields.
\end{abstract}



\section{\bf Introduction}

An important role in quantum control is played by performance criteria which address robustness with respect to unmodelled dynamics of quantum systems and their  interaction with uncertain environment (including external fields and other quantum or classical systems). In the case of open quantum harmonic oscillators (OQHOs), used in linear quantum systems theory \cite{NY_2017,P_2017}, the modelling errors may come from imprecise knowledge of the coefficients  of linear quantum stochastic differential equations (QSDEs)  which govern such systems in the framework of the Hudson-Parthasarathy calculus \cite{HP_1984,P_1992,P_2015}. Another source of uncertainties is the imprecisely known quantum state of the external bosonic fields which drive the system and affect its statistical properties such as mean square values and other moments of the system variables \cite{MJ_2012,NJP_2009}.

In the case of a quantum relative entropy \cite{OW_2010} description of the system-field state uncertainty \cite{YB_2009},  the worst-case values of quadratic costs admit upper bounds in terms of a quadratic-exponential functional (QEF) \cite{VPJ_2018b}   (see also \cite{B_1996}). The QEF is the exponential of an  integral-of-quadratic function of system variables, which is an alternative version   of the original quantum risk-sensitive cost \cite{J_2004,J_2005} based on time-ordered exponentials.  In addition to the robust performance bounds,  the QEF is also related to the large deviations of quantum trajectories \cite{VPJ_2018a}, so that its minimization (by varying the admissible controller parameters for a given quantum plant) can be used in order to secure more conservative closed-loop system dynamics. The QEF criteria and the original quantum risk-sensitive costs can share useful properties due to the Lie-algebraic links \cite{VPJ_2019a} between these two classes of cost functionals.

A closed-form representation for exponential moments of quadratic forms in a finite number of quantum variables with canonical commutation relations (CCRs) over Gaussian states \cite{KRP_2010},  obtained in \cite{VPJ_2018c}, employs symplectic geometric techniques and averaging over auxiliary classical random variables. Application of this approach to QEFs for linear quantum stochastic systems requires its extension to an infinite set of quantum variables associated with the Heisenberg evolution of the OQHO over a bounded time interval. Such an extension, proposed in \cite{VPJ_2019b}, uses  a finite-horizon quantum Karhunen-Loeve (QKL) expansion of the system variables over the eigenbasis of their two-point quantum covariance kernel, with the coefficients being
organised as conjugate pairs of noncommuting quantum mechanical positions and momenta \cite{S_1994}. Being associated with the covariance kernel, this QKL expansion is a quantum counterpart of its predecessor for classical random processes \cite{GS_2004} (see also \cite{IT_2010}).
 However, this result leads to a double series in the integral-of-quadratic representation due to the lack of orthogonality between the real and imaginary parts of the eigenfunctions. A single series representation for the QEF is achieved in \cite{VJP_2019} for a class of one-mode OQHOs with positive definite energy matrices by using a modified QKL expansion for the system variables over a different eigenbasis of the two-point CCR kernel with mutually orthogonal real and imaginary parts of the eigenfunctions.

The present paper frees the approach of \cite{VJP_2019} from dependence on the specific features of the one-mode case and extends it to multimode OQHOs. For such systems, we recast the eigenanalysis of the two-point commutator kernel as a boundary value problem (BVP) for a second-order ODE. 
The structure of the kernel leads to eigenfunctions with mutually orthogonal real and imaginary parts. The resulting QKL expansion of the system variables allows the QEF to be represented in terms of a single series over squared position-momentum pairs. We then apply the symplectic factorization and parameter randomization techniques of \cite{VPJ_2018c} in order to express the QEF, up to a correction factor, as an average value of the exponential of a bilinear function of the quantum system variables and an auxiliary  classical Gaussian random process. The structure of this representation resembles that of the Doleans-Dade exponential \cite{DD_1970} in the context of Girsanov's theorem  \cite{G_1960} on absolutely continuous change of measure for classical Ito processes. However, the correction factor originates from the Weyl CCRs \cite{F_1989} rather than the martingale property of the  Radon-Nikodym derivative process.

The Girsanov type representation, which involves only the commutation structure of the quantum dynamics regardless of the system-field state,  allows the QEF to be related to the moment-generating functional (MGF) of the system variables. This result is then specified for the invariant multipoint Gaussian quantum state \cite{VPJ_2018a} when the oscillator is driven by vacuum fields. In this case, the QEF is reduced to a quadratic-exponential moment for the auxiliary Gaussian random process.

In addition to their relevance to quantum risk-sensitive control, these results can contribute to the study of operator exponential structures which are being actively researched in mathematical physics and quantum probability (for example, in the context of operator algebras \cite{AB_2018}, moment-generating functions for quadratic Hamiltonians \cite{PS_2015} and the quantum L\'{e}vy area \cite{CH_2013,H_2018}).

The paper is organised as follows.
Section~\ref{sec:sys} specifies the class of linear quantum stochastic systems being considered.
Section~\ref{sec:comm} recasts the integral operator with the two-point CCR kernel of the system variables as a BVP. 
Section~\ref{sec:eig} describes the eigenvalues and eigenfunctions for this operator. 
Section~\ref{sec:QKL} employs this eigenbasis for a finite-horizon QKL expansion of the system variables.
Section~\ref{sec:quadro} develops a Girsanov type representation for QEFs and specifies it for the case of vacuum fields.
Section~\ref{sec:conc} provides concluding remarks.
Appendix establishes an auxiliary theorem on a randomised representation for quadratic-exponential functions of position-momentum pairs.

\section{\bf Open quantum harmonic oscillators}
\label{sec:sys}

For what follows, let $W:=(W_k)_{1\< k \< m}$ be a multichannel quantum Wiener process, organised as a column-vector  of an even number of time-varying self-adjoint operators $W_1(t), \ldots, W_m(t)$ on a symmetric Fock space $\fF$ \cite{P_1992}, which represent bosonic fields (the time argument $t$ will often be omitted for brevity). For any $t\>0$ and $k=1, \ldots, m$, the operator $W_k(t)$ acts on a subspace $\fF_t$ of $\fF$, with  the 
increasing family $(\fF_t)_{t\> 0}$  playing the role of a filtration for the Fock space $\fF$  in accordance with its continuous tensor-product structure  \cite{PS_1972}. The component quantum Wiener processes $W_1, \ldots, W_m$ satisfy the two-point CCRs 
$    [W(s), W(t)^\rT]
     := ([W_j(s), W_k(t)])_{1\< j,k\< m}
     =
    2i\min(s,t)J$ 
for all $s,t\>0$, where $(\cdot)^\rT$ is the transpose,    $[\alpha,\beta]:= \alpha\beta - \beta\alpha$ is the commutator of linear operators, and $i:= \sqrt{-1}$ is the imaginary unit. Here, 
    $J:=  \bJ \ox I_{m/2}$ 
is an orthogonal real  antisymmetric matrix (that is, $J^2 = -I_m$), where $\ox$ is the Kronecker product, $I_m$ is the identity matrix of order $m$, and
\begin{equation}
\label{bJ}
\bJ: = {\begin{bmatrix}
        0 & 1\\
        -1 & 0
    \end{bmatrix}}
\end{equation}
spans the subspace of antisymmetric matrices of order 2.  If $\xi$ and $\eta:= -i\d_\xi$ are the quantum mechanical position and momentum operators \cite{S_1994} (implemented on the Schwartz space \cite{V_2002}), then
$[\xi,\eta] = i$, and the vector
\begin{equation}
\label{zeta}
    \zeta:=
    {\begin{bmatrix}
        \xi\\
        \eta
    \end{bmatrix}}
\end{equation}
has the CCR matrix $\frac{1}{2}\bJ$ in the sense that $[\zeta, \zeta^\rT] = i\bJ$. Such conjugate position-momentum pairs provide building blocks for more complicated CCRs between several (or an infinite number of) quantum variables.

Now, consider a multimode open quantum harmonic oscillator (OQHO), which interacts with external bosonic fields and is endowed with an even number of time-varying  self-adjoint quantum variables $X_1(t), \ldots, X_n(t)$ acting on the subspace $\fH_t:= \fH_0 \ox \fF_t$ of the system-field tensor-product space $\fH:= \fH_0 \ox \fF$. Here, $\fH_0$ is a complex separable Hilbert space for the action of the initial system variables $X_1(0), \ldots, X_n(0)$. The vector     $X:=(X_k)_{1\< k\< n}$ of system variables satisfies the Weyl CCRs \cite{F_1989}
\begin{equation}
\label{Weyl}
  \re^{i(u+v)^\rT X(t)}
  =
  \re^{iu^\rT\Theta v}
    \re^{iu^\rT X(t)}
      \re^{iv^\rT X(t)},
      \qquad
      u,v\in \mR^n,
\end{equation}
for any $t\> 0$,
with their infinitesimal Heisenberg form given by
\begin{equation}
\label{XCCR}
    [X(t),X(t)^\rT] = 2i \Theta,
\end{equation}
where $\Theta$ is a constant real antisymmetric matrix of order $n$, which is assumed to be nonsingular.  The evolution of the OQHO is governed by a linear QSDE
\begin{equation}
\label{dX}
    \rd X = AX \rd t + B\rd W
\end{equation}
driven by the quantum Wiener process $W$ described above. Here, the matrices $A \in \mR^{n\x n}$, $B\in \mR^{n\x m}$ are parameterised as
\begin{equation}
\label{AB}
    A = 2\Theta (R + M^\rT JM),
     \qquad
     B = 2\Theta M^\rT
\end{equation}
by the energy and coupling matrices $R = R^\rT \in \mR^{n\x n}$, $M \in \mR^{m\x n}$ which specify the system Hamiltonian $\frac{1}{2} X^\rT R X$ and the vector $MX$ of $m$ system-field coupling operators. Due to their specific structure (\ref{AB}), the matrices $A$, $B$ satisfy the physical realizability (PR) condition \cite{JNP_2008}
\begin{equation}
\label{PR}
    A \Theta + \Theta A^\rT +  \mho = 0,
    \qquad
    \mho := BJB^\rT,
\end{equation}
which is an algebraic manifestation of the CCR preservation in (\ref{XCCR}). Since (\ref{PR}) is organised as an algebraic Lyapunov equation (ALE) with respect to $\Theta$, the CCR matrix can be recovered uniquely as $\Theta = \int_0^{+\infty} \re^{tA} \mho \re^{tA^\rT} \rd t$ in the case when $A$ is Hurwitz, which is assumed in what follows.   The OQHO retains its general structure under the linear transformations
\begin{equation}
\label{XSX}
    X \mapsto SX
\end{equation}
of the system variables with arbitrary nonsingular matrices $S \in \mR^{n\x n}$, except that its matrices in (\ref{XCCR})--(\ref{AB}) are modified as
\begin{equation}
\label{new}
  \Theta\mapsto S\Theta S^\rT,
  \qquad
  R\mapsto S^{-\rT}RS^{-1},
  \qquad
  M\mapsto M S^{-1}  ,
  \qquad
  A\mapsto SAS^{-1},
  \qquad
  B \mapsto SB,
\end{equation}
with the CCR matrix remaining nonsingular, and the dynamics matrix remaining Hurwitz.

\section{\bf Integral operator with the 
commutator kernel}
\label{sec:comm}

Due to linearity of the QSDE (\ref{dX}), the system variables satisfy
\begin{equation}
\label{Xsol}
    X(t)
    =
    \re^{(t-s)A} X(s)
    +
    \int_s^t
    \re^{(t-v)A} B \rd W(v),
    \qquad
    t\> s\> 0.
\end{equation}
In combination with the commutativity $[\rd W(v), X(s)^\rT] = 0$  between the forward increments of the quantum Wiener   process and the past system variables for all $v\> s\>0$, and the CCRs (\ref{XCCR}),  the relation (\ref{Xsol}) yields the two-point CCRs \cite{VPJ_2018a}:
\begin{equation}
\label{XXcomm}
    [X(s), X(t)^\rT]
    =
    2i\Lambda(s-t),
    \qquad
    s,t\>0,
\end{equation}
with
\begin{equation}
\label{Lambda}
    \Lambda(\tau)
     :=
    \left\{
    {\begin{matrix}
    \re^{\tau A}\Theta& {\rm if}\  \tau\> 0\\
    \Theta\re^{-\tau A^{\rT}} & {\rm if}\  \tau< 0\\
    \end{matrix}}
    \right.
    =
    -\Lambda(-\tau)^\rT.
\end{equation}
The one-point CCRs (\ref{XCCR}) are a particular case of (\ref{XXcomm}), (\ref{Lambda}) since $\Lambda(0) = \Theta$. For a fixed but otherwise arbitrary time horizon $T>0$,   the two-point CCR function $\Lambda$ gives rise to a skew self-adjoint linear operator $\cL$ (in view of the second equality in (\ref{Lambda})) which maps a function $f \in L^2([0,T],\mC^n)$ to another such function $g:= \cL(f)$ as
\begin{equation}
\label{fg}
    g(s)
    :=
    \int_0^T
    \Lambda(s-t) f(t)
    \rd t,
    \qquad
    0\< s \< T.
\end{equation}
Here, the Hilbert space $L^2([0,T],\mC^n)$ of square integrable $\mC^n$-valued functions on $[0,T]$ is endowed with the inner product
$
    \bra f,g\ket := \int_0^T f(t)^* g(t)\rd t
$ and the norm $\|f\|:= \sqrt{\bra f,f\ket} = \sqrt{\int_0^T |f(t)|^2\rd t}$, where $(\cdot)^*:= {{\overline{(\cdot)}}}^\rT$ is the complex conjugate transpose.
The following theorem represents the integral operator in the form of  a BVP.

\begin{thm}
\label{thm:BVP}
Suppose the matrix $\mho = -\mho^\rT  \in \mR^{n\x n}$ in (\ref{PR}) satisfies
\begin{equation}
\label{BJB}
  \det \mho \ne 0.
\end{equation}
Then the integral operator $\cL$ in (\ref{fg}) with the two-point CCR function (\ref{Lambda}) describes the solution of the second-order ODE
\begin{equation}
\label{ODE}
  g'' + (\mho A^\rT \mho^{-1}-A) g'- \mho A^\rT \mho^{-1} A g   = -\mho f
\end{equation}
over the time interval $[0,T]$ subject to the boundary conditions
\begin{equation}
\label{g+0_g-T}
g'(0)  = -\Theta A^\rT \Theta^{-1} g(0),
\qquad
g'(T)  = A g(T).
\end{equation}
In particular, this operator does not contain $0$ among its eigenvalues. \hfill$\square$
\end{thm}
\begin{proof}
In view of the structure of the kernel function (\ref{Lambda}), the operator $\cL$ in (\ref{fg}) acts from $L^2([0,T],\mC^n)$ to the subspace of twice differentiable functions with square integrable second-order derivatives. Indeed, similarly to the Wiener-Hopf method, the image $g=\cL(f)$ is represented as
\begin{equation}
\label{ggg}
  g(s) = g_+(s) + \Theta g_-(s),
  \qquad
  0\< s\< T,
\end{equation}
where
\begin{equation}
\label{g+-}
    g_+(s)
    :=
    \int_0^s
    \Lambda(s-t)
    f(t)\rd t,
    \qquad
    g_-(s)
    :=
    \int_s^T
    \re^{(t-s)A^\rT}
    f(t)\rd t
\end{equation}
are absolutely continuous functions satisfying the boundary conditions
\begin{equation}
\label{g+-0}
    g_+(0) = g_-(T) = 0.
\end{equation}
By using the matrix exponential structure of (\ref{Lambda}) again, differentiation of (\ref{g+-})  yields
\begin{equation}
\label{g+-'}
    g_+' = \Theta f + Ag_+,
    \qquad
    g_-' = -f - A^\rT g_-,
\end{equation}
and hence, in view of (\ref{ggg}), the function
\begin{equation}
\label{g'}
    g'
    =
    g_+' + \Theta g_-'
    = \Theta f + Ag_+ - \Theta (f + A^\rT g_-)
    =
    Ag_+ - \Theta  A^\rT g_-
\end{equation}
is also absolutely continuous.
The relations (\ref{ggg}), (\ref{g'}) can be represented as
\begin{equation}
\label{set1}
    {\begin{bmatrix}
      I_n & \Theta\\
      A & -\Theta A^\rT
    \end{bmatrix}
    \begin{bmatrix}
      g_+\\
      g_-
    \end{bmatrix}}
    =
    {\begin{bmatrix}
      g\\
      g'
    \end{bmatrix}}.
\end{equation}
The matrix of this set of linear equations is nonsingular since
$
    \det
    {\small\begin{bmatrix}
      I_n & \Theta\\
      A & -\Theta A^\rT
    \end{bmatrix}}
    =
    \det(-\Theta A^\rT - A\Theta) = \det \mho \ne 0
$,
where the Schur complement \cite{HJ_2007} of the block $I_n$ is used together with the PR condition (\ref{PR}) and the assumption (\ref{BJB}). Calculation of the inverse of this matrix  allows (\ref{set1}) to be solved as
\begin{equation}
\label{set2}
    {\begin{bmatrix}
      g_+\\
      g_-
    \end{bmatrix}}
    =
    {\begin{bmatrix}
      I_n + \Theta \mho^{-1} A & -\Theta\mho^{-1}\\
      -\mho^{-1}A & \mho^{-1}
    \end{bmatrix}
    \begin{bmatrix}
      g\\
      g'
    \end{bmatrix}}.
\end{equation}
By differentiating (\ref{g'}) and using (\ref{g+-'}), (\ref{PR}), (\ref{set2}), it follows that
\begin{align}
\nonumber
    g''
    & =
    Ag_+' - \Theta  A^\rT g_-'
    =
    A(\Theta f + Ag_+) - \Theta  A^\rT (-f - A^\rT g_-)\\
\nonumber
    & =
    (A\Theta + \Theta  A^\rT) f + A^2 g_+ + \Theta  (A^\rT)^2 g_-\\
\nonumber
    & =
    -\mho f
    +
    {\begin{bmatrix}
      A^2 &  \Theta  (A^\rT)^2
    \end{bmatrix}
    \begin{bmatrix}
      g_+\\
      g_-
    \end{bmatrix}}    \\
\nonumber
    & =
    -\mho f
    +
    {\begin{bmatrix}
      A^2 &  \Theta  (A^\rT)^2
    \end{bmatrix}
    \begin{bmatrix}
      I_n + \Theta \mho^{-1} A & -\Theta\mho^{-1}\\
      -\mho^{-1}A & \mho^{-1}
    \end{bmatrix}
    \begin{bmatrix}
      g\\
      g'
    \end{bmatrix}}\\
\label{g''}
    & =
    -\mho f + \mho A^\rT \mho^{-1} A g  + (A-\mho A^\rT \mho^{-1}) g',
\end{align}
where use is also made of the identity
$
    A^2 \Theta - \Theta (A^\rT)^2 =
    A A\Theta - \Theta A^\rT A^\rT
    =
    -A(\mho + \Theta A^\rT) + (\mho + A\Theta )A^\rT
    =
    \mho A^\rT - A \mho
$, which follows from the PR condition (\ref{PR}). The ODE (\ref{g''}) is identical to (\ref{ODE}), while the boundary conditions (\ref{g+0_g-T}) are obtained from (\ref{g+-0}) by using (\ref{set2}) along with the identity $A + \mho \Theta^{-1} = (A\Theta + \mho) \Theta^{-1} = -\Theta A^\rT \Theta^{-1}$ which follows from (\ref{PR}). If $0$ were an eigenvalue of the operator (\ref{fg}), then there would exist a function $f$, not identically zero and satisfying (\ref{ODE}) with $g=0$, so that $\mho f = 0$ would hold almost everywhere  on the time interval $[0,T]$, thus leading to a contradiction   because of (\ref{BJB}).
\end{proof}

Since the CCR matrix $\Theta$ is assumed to be nonsingular, (\ref{BJB}) is equivalent to
\begin{equation}
\label{MJM}
  \det (M^\rT J M) \ne 0
\end{equation}
in view of the parameterisation of the matrix $B$ in (\ref{AB}) in terms of the coupling matrix $M$. A necessary condition for (\ref{MJM}) is that $M$ is of full column rank: $M^\rT M \succ 0$, and hence, $m\> n$ (that is, the number of oscillator modes does not exceed the number of external field channels).

Theorem~\ref{thm:BVP} characterizes the kernel $\Lambda$ in (\ref{Lambda}) as the Green function for the BVP (\ref{ODE}), (\ref{g+0_g-T}). More precisely, the ODE (\ref{ODE}) is equivalent to
\begin{equation}
\label{ODE1_F}
    {\begin{bmatrix}
      g\\
      g'
    \end{bmatrix}}'
    =
    F
    {\begin{bmatrix}
      g\\
      g'
    \end{bmatrix}}
    -
    {\begin{bmatrix}
      0\\
      \mho
    \end{bmatrix}}
    f,
    \qquad
    F:=
    {\begin{bmatrix}
      0 &  I_n\\
      \mho A^\rT \mho^{-1} A & A - \mho A^\rT \mho^{-1}
    \end{bmatrix}}.
\end{equation}
Hence, the solution of the BVP (\ref{ODE}), (\ref{g+0_g-T}) is given by
\begin{equation}
\label{gsol}
    g(s)
    =
    \begin{bmatrix}
      I_n & 0
    \end{bmatrix}
    \Big(
    \re^{sF}
    V
    g(0)
    -
    \int_0^s
    \re^{(s-t)F}
    {\begin{bmatrix}
      0\\
      \mho
    \end{bmatrix}}
    f(t)
    \rd t
    \Big)
\end{equation}
for all $0\< s \< T$,
where $g(0)$ is found uniquely from the terminal condition in (\ref{g+0_g-T}) as
\begin{equation}
\label{g0}
    g(0)
    =
    G(T)^{-1}
    U
    \int_0^T
    \re^{(T-t)F}
    {\begin{bmatrix}
      0\\
      \mho
    \end{bmatrix}}
    f(t)
    \rd t,
\end{equation}
with
\begin{equation}
\label{G}
    G(T):=
    U
    \re^{TF}
    V.
\end{equation}
Here, the matrices
\begin{equation}
\label{UV}
    U
    :=
        {\begin{bmatrix}
      A & - I_n
    \end{bmatrix}},
    \qquad
    V
    :=
    {\begin{bmatrix}
      I_n\\
      -\Theta A^\rT \Theta^{-1}
    \end{bmatrix}}
\end{equation}
come from the terminal and initial conditions in (\ref{g+0_g-T}), respectively.
Substitution of (\ref{g0}) into (\ref{gsol}) and comparison of the result with (\ref{fg}) lead to
\begin{equation}
\label{Lam}
    \Lambda(s-t)
    =
    {\begin{bmatrix}
      I_n & 0
    \end{bmatrix}}
    \big(
    \re^{sF}
    V
    G(T)^{-1}
    U
    \re^{(T-t)F}
    -
    \chi_{[0,s]}(t)
    \re^{(s-t)F}
    \big)
    {\begin{bmatrix}
      0\\
      \mho
    \end{bmatrix}}
\end{equation}
for all $0\< s, t\< T$, where $\chi_C(\cdot)$ is the indicator function for a set $C$. The relation (\ref{g0}) and its corollary  (\ref{Lam}) use
the nonsingularity of the matrix $G(T) \in \mR^{n\x n}$ in (\ref{G}), which is discussed below.

\begin{lem}
\label{lem:Ggood}
  Under the condition (\ref{BJB}), the matrix $G(T)$ in (\ref{G}) is nonsingular for any time horizon $T\>0$. \hfill$\square$
\end{lem}
\begin{proof}
The matrices $F$, $U$ in (\ref{ODE1_F}), (\ref{UV}) satisfy
$    UF =
    \begin{bmatrix}
      -\mho A^\rT \mho^{-1} A &
      \mho A^\rT \mho^{-1}
    \end{bmatrix}
    =
    -\mho A^\rT \mho^{-1} U
$.
In combination with (\ref{G}),
this identity leads to the linear ODE
$
    G'(T) = UF\re^{TF}V = -\mho A^\rT \mho^{-1} G(T)
$
for any $T\>0$, and hence,
$
    \det G(T) = \re^{-T\Tr(\mho A^\rT \mho^{-1})} G(0) = \re^{-T\Tr A} G(0)
$
due to the invariance of the matrix trace under similarity transformations. Therefore, the matrix $G(T)$ inherits nonsingularity from
$
  G(0) = UV = A+\Theta A^\rT \Theta^{-1} = (A\Theta + \Theta A^\rT)\Theta^{-1} = -\mho \Theta^{-1}
$,
where (\ref{UV}), (\ref{PR}) are used together with (\ref{BJB}) and the condition $\det \Theta \ne 0$.
\end{proof}

Since the kernel (\ref{Lambda}) is continuous,
the image of any bounded subset of $L^2([0,T], \mC^n)$ under the operator $\cL$ in  (\ref{fg}) consists of uniformly bounded equicontinuous functions.\footnote{Upper bounds for the functions $|g|$,  $|g'|$ in (\ref{fg}) in terms of $\|f\|$ can also be obtained by using the proof of Theorem~\ref{thm:BVP}.} By the Arzela-Ascoli theorem \cite{RS_1980}, this implies relative compactness of the image set in the sense of the uniform norm and hence, in $L^2([0,T], \mC^n)$, whereby $\cL$ is compact.

\section{\bf Eigenbasis for the two-point commutator kernel}
\label{sec:eig}

Since the operator $\cL$ in (\ref{fg}) is skew self-adjoint, all  its eigenvalues are purely imaginary. They are symmetric about the origin, since  for any eigenfunction $f: [0,T]\to \mC^n$  of this operator with an eigenvalue $i\omega$, where $\omega\in \mR$ (so that $\cL(f) = i\omega f$), the function $\overline{f}$ satisfies $\cL(\overline{f}) = -i\omega \overline{f}$ and is, therefore,  an eigenfunction with the eigenvalue $-i\omega$.
Now, let $f$ be an eigenfunction of $\cL$ with the eigenvalue $i\omega$, where $\omega >0$ without loss of generality under the condition (\ref{BJB}) will be referred to as an eigenfrequency of $\cL$.  Substitution of  $g:= i\omega f$ into (\ref{ODE}), (\ref{g+0_g-T}), represents the eigenvalue problem as the  BVP
\begin{equation}
\label{ODEeig}
  f'' + (\mho A^\rT \mho^{-1}-A) f'- \mho A^\rT \mho^{-1} A f   = \frac{i}{\omega}\mho f
\end{equation}
over the time interval $[0,T]$ with the boundary conditions
\begin{equation}
\label{g+0_g-Teig}
f'(0)  = -\Theta A^\rT \Theta^{-1} f(0),
\qquad
f'(T)  = A f(T).
\end{equation}
In accordance with (\ref{ODE1_F}), the ODE (\ref{ODEeig}) is represented as
\begin{equation}
\label{ODE2}
    {\begin{bmatrix}
      f\\
      f'
    \end{bmatrix}}'
    =
    D(\omega)
    {\begin{bmatrix}
      f\\
      f'
    \end{bmatrix}},
\end{equation}
where
\begin{equation}
\label{D}
    D(\omega)
    :=
    F
    +
    \tfrac{i}{\omega}
    {\begin{bmatrix}
      0 &  0\\
      \mho & 0
    \end{bmatrix}}
\end{equation}
is an appropriate modification of the matrix $F$ in (\ref{ODE1_F}).
By using the matrix exponential for the solution of (\ref{ODE2}) with the initial condition
\begin{equation}
\label{f0}
    {\begin{bmatrix}
      f(0)\\
      f'(0)
    \end{bmatrix}}
    =
    V
    f(0),
\end{equation}
with $V$ from (\ref{UV}),
it follows that the terminal condition in (\ref{g+0_g-Teig}) is equivalent to
\begin{equation}
\label{Ef}
    E(\omega)
    f(0) = 0,
\end{equation}
with
\begin{equation}
\label{E}
    E(\omega)
    :=
    U
    \re^{T D(\omega)}
    V.
\end{equation}
The existence of $f(0) \in \ker E(\omega)\setminus \{0\}$ in (\ref{Ef}) is equivalent to 
 $    \det E(\omega) = 0$ 
for the eigenvalues $i\omega$ of the operator $\cL$ in (\ref{fg}). Note that $\lim_{\omega \to +\infty} \det E(\omega) = \det G(T)\ne 0$ in view of (\ref{E}), (\ref{D}), (\ref{G}) and Lemma~\ref{lem:Ggood}, in accordance with the boundedness of the operator $\cL$.
Any $f(0)$ in (\ref{Ef}) gives rise to the eigenfunction
\begin{equation}
\label{ff}
    f(t) =
    {\begin{bmatrix}
        I_n & 0
    \end{bmatrix}}
    \re^{tD(\omega)}
    V
    f(0),
    \quad
    0\< t \< T,
\end{equation}
in view of (\ref{f0}) (with $\dim \ker E(\omega)$ being the dimension of the corresponding subspace of eigenfunctions in $L^2([0,T],\mC^n)$).  The functions
$
    \varphi:= \Re f$ and
$
    \psi := \Im f
$,
associated with (\ref{ff}) (with the real and imaginary parts being taken entrywise),
belong to the real subspace $L^2([0,T], \mR^n)$  and are related by
\begin{equation}
\label{pairs}
    \cL(\varphi) = -\omega \psi,
    \qquad
    \cL(\psi) = \omega \varphi.
\end{equation}
Since $\bra\cdot, \cdot\ket$ is symmetric on $L^2([0,T], \mR^n)$, and the operator $\cL$ in (\ref{fg}) is skew self-adjoint, then, in view of $\omega \ne 0$, it follows from (\ref{pairs})  that
$
    \|\varphi\| = \|\psi\|$ and
$
    \bra \varphi, \psi\ket = 0
$.
This property is equivalent to $\bra \overline{f},f\ket  = 0$. Moreover, if $f$, $h$ are two eigenfunctions with eigenvalues $i\omega$ and $i\mu$, so that $\cL(f) = i\omega f$ and $\cL(h) = i\mu h$, then $\omega\ne \mu$ implies $\bra f,h\ket = 0$, while the fulfillment of $\omega >0$ and $\mu>0$ implies $\bra \overline{f}, h\ket = 0$. The eigenfunctions with a common eigenvalue are orthonormalised by the Gram-Schmidt procedure \cite{RS_1980}. The resulting eigenfunctions
\begin{equation}
\label{eigff}
    f_k = \varphi_k + i\psi_k,
    \qquad
    \overline{f_k} = \varphi_k - i\psi_k
\end{equation}
of the operator $\cL$ in (\ref{fg})
with
\begin{equation}
\label{reim}
  \varphi_k := \Re f_k,
  \qquad
  \psi_k := \Im f_k,
\end{equation}
and the corresponding eigenvalues $\pm i \omega_k $ and eigenfrequencies $\omega_k>0$ (labeled by positive integers $k \in \mN:= \{1,2,3,\ldots\}$),  satisfy
\begin{equation}
\label{ffff}
    \bra f_j,f_k\ket = \delta_{jk},
    \qquad
    \bra \overline{f_j},f_k\ket = 0,
\end{equation}
or equivalently,
\begin{equation}
\label{reimort}
    \bra \varphi_j,\varphi_k\ket
    =
    \bra \psi_j,\psi_k\ket
    =
    \tfrac{1}{2}\delta_{jk},
    \qquad
    \bra \varphi_j,\psi_k\ket = 0
\end{equation}
for all  $
    j,k\in \mN$, with $\delta_{jk}$ the Kronecker delta. An equivalent form of (\ref{reimort}) is
\begin{equation}
\label{hh}
  \int_0^T
  h_j(t)^\rT h_k(t)
  \rd t
  =
  {\begin{bmatrix}
    \bra \varphi_j, \varphi_k\ket & \bra \varphi_j, \psi_k\ket\\
    \bra \psi_j, \varphi_k\ket & \bra \psi_j, \psi_k\ket
  \end{bmatrix}}
  =
  \tfrac{1}{2}\delta_{jk}
  I_2,
\end{equation}
where the functions $h_k \in L^2([0,T], \mR^{n\x 2})$ are formed from the $\mR^n$-valued functions in (\ref{reim}) as
\begin{equation}
\label{hk}
    h_k:=
    {\begin{bmatrix}
      \varphi_k &
      \psi_k
    \end{bmatrix}}.
\end{equation}
The kernel
(\ref{Lambda}) admits the following expansion over the orthonormal eigenfunctions (\ref{eigff}):
\begin{align}
\nonumber
    \Lambda(s-t)
    & =
    i
    \sum_{k=1}^{+\infty}
    \omega_k
    (f_k(s)f_k(t)^*-\overline{f_k(s)}f_k(t)^\rT)\\
\nonumber
    & = -2
    \sum_{k=1}^{+\infty}
    \omega_k
    \Im(f_k(s)f_k(t)^*),\\
\nonumber
    & =
    2
    \sum_{k=1}^{+\infty}
    \omega_k
    (\varphi_k(s)\psi_k(t)^\rT-\psi_k(s)\varphi_k(t)^\rT)\\
\label{Lambdaseries}
    & =
    2
    \sum_{k=1}^{+\infty}
    \omega_k
    h_k(s) \bJ h_k(t)^\rT,
    \qquad
    0\< s,t\< T,
\end{align}
where the matrix $\bJ$ is given by (\ref{bJ}). This expansion
is similar to the Mercer representation \cite{KZPS_1976} of positive semi-definite self-adjoint kernels.  Since $-i\cL$ is a self-adjoint operator on $L^2([0,T], \mC^n)$ with the kernel $-i\Lambda$ and the corresponding eigenvalues $\pm \omega_k$, then, in accordance with (\ref{Lambdaseries}), (\ref{hh}), the squared  Hilbert-Schmidt norm \cite{RS_1980}  of the operator $\cL$ in (\ref{fg}) is
\begin{equation}
\label{HSnorm}
    \|\cL\|_{\rHS}^2
    =
    \int_{[0,T]^2}
    \|\Lambda(s-t)\|_\rF^2
    \rd s\rd t
    =
    2\sum_{k=1}^{+\infty}
    \omega_k^2
    =
    -
    \Tr (\cL^2),
\end{equation}
where $\|N\|_\rF := \sqrt{\Tr (N^*N)}$ is the Frobenius norm of a real or complex matrix $N$. In accordance with its trace in (\ref{HSnorm}), the positive definite self-adjoint operator $-\cL^2$ has the eigenvalues $\omega_k^2$ of multiplicity $2$.

 \section{\bf Quantum Karhunen-Loeve expansion using the two-point commutator kernel eigenbasis}
\label{sec:QKL}

Under the assumption of Theorem~\ref{thm:BVP}, 
with the system variables of the OQHO over the time interval $[0,T]$, we associate a sequence of quantum variables
\begin{equation}
\label{gammak}
  \gamma_k
  :=
  \tfrac{1}{\sqrt{\omega_k}}
  \int_0^T
  f_k(t)^\rT
  X(t)
  \rd t
  =
  \xi_k + i\eta_k
\end{equation}
on the system-field space $\fH$. Here,
the eigenfunctions (\ref{eigff}) of the two-point CCR kernel (\ref{Lambda}) are used together with (\ref{reim}) and the eigenfrequencies $\omega_k$, and
\begin{align}
\label{xik}
    \xi_k
    & :=
    \Re
    \gamma_k
    =
  \tfrac{1}{\sqrt{\omega_k}}
  \int_0^T
  \varphi_k(t)^\rT
  X(t)
  \rd t,\\
\label{etak}
    \eta_k
    & :=
    \Im
    \gamma_k
    =
  \tfrac{1}{\sqrt{\omega_k}}
  \int_0^T
  \psi_k(t)^\rT
  X(t)
  \rd t,
\end{align}
where the real and imaginary parts are extended from $\mC$ 
to quantum variables as $\Re z := \frac{1}{2} (z + z^\dagger)$, $\Im z := \frac{1}{2i} (z - z^\dagger)$, with $(\cdot)^\dagger$ the operator adjoint. The quantum variables (\ref{gammak}) satisfy
\begin{align}
\nonumber
    [\gamma_j,\gamma_k^\dagger]
    &=
  \tfrac{1}{\sqrt{\omega_j\omega_k}}
  \int_{[0,T]^2}
  f_j(s)^\rT
  [X(s),X(t)^\rT]
  \overline{f_k(t)}
  \rd s\rd t\\
\nonumber
  & =
  \tfrac{2i}{\sqrt{\omega_j\omega_k}}
  \int_{[0,T]^2}
  f_j(s)^\rT
  \Lambda(s-t)
  \overline{f_k(t)}
  \rd s\rd t\\
\nonumber
  & =
  \tfrac{2}{\sqrt{\omega_j\omega_k}}
  \int_{[0,T]^2}
  f_j(s)^\rT
    \sum_{\ell=1}^{+\infty}
    \omega_\ell
    (\overline{f_\ell(s)}f_\ell(t)^\rT
    -
    f_\ell(s)f_\ell(t)^*)
  \overline{f_k(t)}
  \rd s\rd t\\
\nonumber
  & =
  \tfrac{2}{\sqrt{\omega_j\omega_k}}
    \sum_{\ell=1}^{+\infty}
    \omega_\ell
    (
        \bra \overline{f_j}, \overline{f_\ell}\ket
        \bra \overline{f_\ell}, \overline{f_k}\ket
        -
        \bra \overline{f_j}, f_\ell\ket
        \bra f_\ell, \overline{f_k}\ket
    )\\
\label{gg+}
  & =
  \tfrac{2}{\sqrt{\omega_j\omega_k}}
    \sum_{\ell=1}^{+\infty}
    \omega_\ell
        \overline{\bra f_j, f_\ell\ket
        \bra f_\ell, f_k\ket}
        =
        2\delta_{jk},
        \qquad
        j,k\in \mN,
\end{align}
where use is made of (\ref{XXcomm}), (\ref{Lambdaseries}),  (\ref{ffff}).
By a similar reasoning,
\begin{equation}
\label{gg}
    [\gamma_j,\gamma_k]
    =
  \tfrac{2}{\sqrt{\omega_j\omega_k}}
    \sum_{\ell=1}^{+\infty}
    \omega_\ell
    (
        \bra \overline{f_j}, \overline{f_\ell}\ket
        \bra \overline{f_\ell}, f_k\ket
        -
        \bra \overline{f_j}, f_\ell\ket
        \bra f_\ell, f_k\ket
    )
    =
    0,
\end{equation}
and hence,
\begin{equation}
\label{g+g+}
    [\gamma_j^\dagger,\gamma_k^\dagger]
    =
    -
    [\gamma_j,\gamma_k]^\dagger
    =0
\end{equation}
for all $j,k\in \mN$.  The CCRs (\ref{gg+})--(\ref{g+g+}) show that 
$\gamma_k$ in (\ref{gammak}) are organised as pairwise commuting annihilation operators (with $\gamma_k^\dagger = \xi_k-i\eta_k$ the corresponding  creation operators), so that the self-adjoint quantum variables $\xi_k$, $\eta_k$ in (\ref{xik}), (\ref{etak}) are conjugate pairs of the quantum mechanical  positions and momenta mentioned in Section~\ref{sec:sys}, with
\begin{equation}
\label{xietacomm}
  [\xi_j,\xi_k] = 0,
  \qquad
  [\eta_j,\eta_k] = 0,
  \qquad
  [\xi_j,\eta_k] = i \delta_{jk}
\end{equation}
for all $j,k\in \mN$.
Accordingly, the vectors
\begin{equation}
\label{zetak}
  \zeta_k
  :=
  {\begin{bmatrix}
    \xi_k\\
    \eta_k
  \end{bmatrix}}
\end{equation}
commute between each other and have a common CCR matrix $\frac{1}{2}\bJ$, since 
$    [\zeta_j,\zeta_k^\rT]
    =
    i \delta_{jk}
    \bJ$, 
where (\ref{xietacomm}) is used together with (\ref{bJ}). In view of the orthonormality of the  eigenbasis (\ref{eigff}), the system variables can be represented in terms of
(\ref{gammak}) as
\begin{align}
\nonumber
    X(t)
    &=
    \sum_{k=1}^{+\infty}
    \sqrt{\omega_k}
    (\overline{f_k(t)}
    \gamma_k
    +
    f_k(t)
    \gamma_k^\dagger
    )
     =
    2
    \sum_{k=1}^{+\infty}
    \sqrt{\omega_k}
    \Re
    (    f_k(t)
    \gamma_k^\dagger
    )\\
\label{XQKL}
    & =
    2
    \sum_{k=1}^{+\infty}
    \sqrt{\omega_k}
    (\varphi_k(t)\xi_k + \psi_k(t)\eta_k)
    =
    2
    \sum_{k=1}^{+\infty}
    \sqrt{\omega_k}
    h_k(t)
    \zeta_k
\end{align}
for all     $0\< t\< T$,
where the coefficients $\sqrt{\omega_k}\zeta_k$ involve the pairs of noncommuting positions and momenta $\xi_k$, $\eta_k$ from (\ref{zetak}), and (\ref{hk}) is used. Extending the results of \cite{VJP_2019} to the multimode case, (\ref{XQKL}) is a QKL 
expansion of the system variables over the eigenbasis of their two-point CCR kernel.

\section{\bf Quadratic-exponential functionals for system variables}
\label{sec:quadro}

Consider an integral-of-quadratic function 
of the system variables over the time interval $[0,T]$   
given by a positive semi-definite self-adjoint quantum variable
\begin{equation}
\label{Q}
    Q
    :=
    \int_0^T
    X(t)^{\rT} X(t)
    \rd t
    =
    \int_0^T
    \sum_{k=1}^n
    X_k(t)^2
    \rd t.
\end{equation}
A more general form
$\int_0^T X(t)^{\rT} \Pi X(t) \rd t$ of such a function, specified by a real positive  definite symmetric matrix $\Pi$ of order $n$, is reduced to (\ref{Q}) by transforming the OQHO as in (\ref{XSX}), (\ref{new}) with $S:= \sqrt{\Pi}$. The performance criteria in linear quadratic Gaussian control and filtering problems \cite{MJ_2012,NJP_2009,ZJ_2012} are concerned with the minimization of quadratic costs $\bE Q := \Tr(\rho Q)$, where $\bE(\cdot)$ is the quantum expectation over an underlying density operator $\rho$ on the system-field space $\fH$. However, imposition of an exponential  penalty on $Q$ in (\ref{Q}) through the QEF 
\cite{VPJ_2018a}
\begin{equation}
\label{Xi}
    \Xi
    :=
    \bE \re^{\frac{\theta}{2} Q},
\end{equation}
as a risk-sensitive cost to be minimised (with $\theta>0$ quantifying the risk sensitivity), leads to  additional properties of the quantum system (the $\tfrac{1}{2}$ factor in (\ref{Xi}) is introduced for convenience). These properties  pertain to the large deviations of quantum trajectories \cite{VPJ_2018a} and robustness to quantum statistical uncertainties with a relative entropy description  \cite{OW_2010,YB_2009}   with respect to the nominal system-field state \cite{VPJ_2018b}.

Consider a representation of the quadratic-exponential function $\re^{\frac{\theta}{2}Q}$ of the system variables in (\ref{Xi})   in a form which is valid irrespective of a particular quantum state $\rho$.  To this end,
by substituting (\ref{XQKL}) into (\ref{Q}) and using (\ref{hh}), it follows that
\begin{align}
\nonumber
      Q
    & =
    4
    \sum_{j,k=1}^{+\infty}
    \sqrt{\omega_j\omega_k}
    \zeta_j^\rT
    \int_0^T
    h_j(t)^\rT
    h_k(t)
    \rd t
    \zeta_k\\
\label{Q1}
    & = 2
    \sum_{k=1}^{+\infty}
    \omega_k
    \zeta_k^\rT
    \zeta_k
    =
    2
    \sum_{k=1}^{+\infty}
    \omega_k
    (\xi_k^2 + \eta_k^2).
\end{align}
The reduction to a single series for $Q$ has been made possible by the mutual orthogonality of the real and imaginary parts of the eigenfunctions  in (\ref{reimort}). A combination of (\ref{Q1}) with
Theorem~\ref{th:fact} of Appendix leads to
\begin{align}
\nonumber
    \re^{\frac{\theta}{2}Q}
    & =
    \prod_{k=1}^{+\infty}
    \re^{\theta \omega_k (\xi_k^2 + \eta_k^2)}
    =
    \prod_{k=1}^{+\infty}
    \big(
    \tfrac{1}{\cosh (\theta \omega_k)}
    \bM
    \re^{\sigma_k(\alpha_k \xi_k + \beta_k \eta_k)}
    \big)\\
\label{eQ}
    & =
    \bM
    \re^{\sum_{k=1}^{+\infty} (\sigma_k(\alpha_k \xi_k + \beta_k \eta_k) - \ln \cosh (\theta\omega_k))}
    =
    \bM \re^{\Sigma-C},
\end{align}
where the commutativity between the position-momentum pairs $(\xi_k,\eta_k)$ for different $k$ is also used.  Here, $\bM(\cdot)$ is the classical expectation over
mutually independent  standard normal (Gaussian with zero mean and unit variance) random variables $\alpha_k$, $\beta_k$, and
\begin{equation}
\label{uvvarsk}
    \sigma_k :=  \sqrt{2\tanh (\theta \omega_k)}
\end{equation}
are related to the eigenfrequencies $\omega_k$ of the operator $\cL$.
Also,
\begin{equation}
\label{Sig}
    \Sigma
    :=
  \sum_{k=1}^{+\infty}
  \sigma_k(\alpha_k \xi_k + \beta_k\eta_k)
\end{equation}
in (\ref{eQ}) is a self-adjoint quantum variable which depends parametrically  (and in a linear fashion) on the classical random variables $\alpha_k$, $\beta_k$. Furthermore,
\begin{equation}
\label{C}
  C
  :=
  \sum_{k=1}^{+\infty}\ln \cosh (\theta \omega_k)
  =
  \sum_{k=1}^{+\infty}\ln \cos(i \theta \omega_k)
  =
  \tfrac{1}{2}
  \Tr \ln\cos (\theta \cL)
\end{equation}
is a nonnegative quantity whose finiteness is secured by that of the Hilbert-Schmidt norm in (\ref{HSnorm}).
Substitution of (\ref{xik}), (\ref{etak}) into (\ref{Sig}) leads to
\begin{equation}
\label{SZ}
    \Sigma
    =
  \sum_{k=1}^{+\infty}
    \tfrac{\sigma_k}{\sqrt{\omega_k}}
    \int_0^T
  (\alpha_k
  \varphi_k(t)
  +
  \beta_k
  \psi_k(t)
  )^\rT
  X(t)
  \rd t
    =
    \sqrt{\theta}
    \int_0^T
    X(t)^\rT
    \rd Z(t),
\end{equation}
where, in view of (\ref{uvvarsk}),
\begin{align}
\nonumber
    Z(t)
    & :=
    \sum_{k=1}^{+\infty}
    \sqrt{2\tanhc (\theta\omega_k)}
    \int_0^t
    (\alpha_k \varphi_k(\tau)
    +
    \beta_k \psi_k(\tau)    )
    \rd \tau\\
\label{Z}
    & =
    \sum_{k=1}^{+\infty}
    \sqrt{\tanhc (\theta\omega_k)}\,
    H_k(t)
    {\begin{bmatrix}
      \alpha_k\\
      \beta_k
    \end{bmatrix}},
    \qquad
    0\< t\< T,
\end{align}
is a classical $\mR^n$-valued Gaussian random process with zero mean. Here,  $H_k \in L^2([0,T], \mR^{n\x 2})$ are auxiliary  functions associated with (\ref{hk}) by
\begin{equation}
\label{Hk}
  H_k(t) := \sqrt{2}\int_0^t h_k(\tau) \rd \tau,
  \qquad
  0\< t \< T,
\end{equation}
and use is also made of the hyperbolic version $\tanhc z := \tanc (-iz)$ of the function $\tanc z := \frac{\tan z}{z}$ (extended to $1$ at $z=0$ by continuity). Since the eigenbasis being used satisfies (\ref{hh}), a combination of (\ref{Hk}) with the Plancherel identity leads to
\begin{equation}
\label{HH}
  \sum_{k=1}^{+\infty}
  H_k(s)H_k(t)^\rT = \min (s,t)I_n,
  \qquad
  0\< s,t\< T,
\end{equation}
which describes the covariance function of a standard Wiener process $V$ in $\mR^n$.  The process $Z$ in (\ref{Z}) has a different covariance function
\begin{equation}
\label{K}
  \bM(Z(s)Z(t)^\rT)
  =
  \sum_{k=1}^{+\infty}
  \tanhc (\theta\omega_k)\,
  H_k(s)H_k(t)^\rT,
  \qquad
  0\< s,t\< T,
\end{equation}
which is majorised by (\ref{HH}) in the sense that
$
    \bM (\bra f, Z\ket^2) 
    \<
    \int_{[0,T]^2}
    \min(s,t)f(s)^\rT f(t)\rd s\rd t
    =
    \bM (\bra f, V\ket^2)
$
for any $f \in L^2([0,T],\mR^n)$ since the function $\tanhc$ does not exceed $1$ on the real axis. Note that $\tanhc (\theta\omega_k)$ are the eigenvalues (of  multiplicity $2$) of the positive definite self-adjoint operator
\begin{equation}
\label{cK}
    \cK
    :=
    \tanhc(i\theta \cL) = \tanc (\theta\cL),
\end{equation}
which is the covariance  operator for the increments of the process $Z$:
\begin{align}
\nonumber
    \bM
    \Big(\int_0^T f(t)^\rT \rd Z(t)\Big)^2
    & =
    \bra f, \cK(f)\ket \\
\label{iso}
    & \<
    \|f\|^2
    =
    \bM\Big(\int_0^T f(t)^\rT \rd V(t)\Big)^2
\end{align}
for any $f \in L^2([0,T], \mR^n)$. Here, the inequality follows from $\cK \preccurlyeq \cI$, with $\cI$ the identity operator on $L^2([0,T], \mR^n)$, while   the last equality in (\ref{iso}) is the Ito isometry.
Reassembling (\ref{C}), (\ref{SZ}) in (\ref{eQ}) leads to
\begin{equation}
\label{eQZ}
  \re^{\frac{\theta}{2}Q}
  =
    \bM \re^{\sqrt{\theta}\int_0^T X(t)^\rT \rd Z(t)  - \frac{1}{2}\Tr \ln \cos (\theta\cL)},
\end{equation}
where the averaging is over the classical Gaussian random process $Z$.
The general structure of this representation is reminiscent of the Doleans-Dade exponential \cite{DD_1970}.  However, the covariance function (\ref{K}) of the process $Z$ in (\ref{eQZ})  is associated with the quantum commutator kernel (\ref{Lambda}), and so also is the correction term $- \frac{1}{2}\Tr \ln \cos (\theta\cL)$. The latter (also see the proof of Theorem~\ref{th:fact}) comes from the Weyl CCRs (rather than the martingale property of the  Radon-Nikodym density process in Girsanov's theorem  \cite{G_1960} on absolutely continuous change of measure for classical Ito processes). In the 
classical case, when the commutator kernel $\Lambda$  vanishes (and so do the linear operator $\cL$ and its eigenfrequencies $\omega_k$), the covariance function (\ref{K}) reduces to (\ref{HH}), and $Z$ becomes the standard Wiener process.

Now, the validity of (\ref{eQZ}) does not rely on a particular system-field quantum state and employs only the commutation structure of the system variables. Nevertheless, its relevance to computing the QEF (\ref{Xi}) 
is clarified by
\begin{equation}
\label{XiZ}
  \Xi
  =
  \re^{- C}
  \bE
    \bM
    \re^{\sqrt{\theta}\int_0^T X(t)^\rT\rd Z(t)}
    =
  \re^{- C}
  \bM
    \Psi(\sqrt{\theta}Z).
\end{equation}
Here,
\begin{equation}
\label{Psi}
    \Psi(z)
    :=
      \bE
    \re^{z(T)^\rT X(T)-\int_0^T z(t)^\rT\rd X(t)}
\end{equation}
is a modified MGF 
for the system variables of the OQHO over the time interval $[0,T]$, considered for 
$z \in L^2([0,T],\mR^n)$ with $z(0)=0$.  In (\ref{XiZ}), use is made of the integration by parts (with $Z(0)=0$), and  the MGF $\Psi$ is evaluated  at the rescaled random process $\sqrt{\theta}Z$ from  (\ref{Z}) in a pathwise fashion. We have also used the commutativity between the quantum and classical expectations $\bE(\cdot)$, $\bM(\cdot)$ which describe the averaging over statistically independent quantum variables and auxiliary classical random variables.

Note that (\ref{XiZ}) does not employ specific assumptions on the system-field state in addition to the finiteness of exponential moments of the system variables (making the MGF (\ref{Psi}) well-defined). We will now specify this relation for the case when the OQHO is driven by vacuum input fields. In this case, since the matrix $A$ in (\ref{AB}) is assumed to be Hurwitz, the system variables have a unique invariant multipoint Gaussian quantum state \cite{VPJ_2018a} with zero mean and the two-point quantum covariance matrix 
$  \bE(X(s)X(t)^\rT)
  =
  P(s-t) + i\Lambda(s-t)$,
$
  s,t\> 0$. 
While its imaginary part (\ref{Lambda}) describes the two-point CCRs (\ref{XXcomm}), the real part
\begin{equation}
\label{P}
    P(\tau)
    =
    \left\{
    {\begin{matrix}
    \re^{\tau A}P_0& {\rm if}\  \tau\> 0\\
    P_0\re^{-\tau A^{\rT}} & {\rm if}\  \tau< 0\\
    \end{matrix}}
    \right.
    =
    P(-\tau)^\rT
\end{equation}
is a positive semi-definite kernel which specifies the two-point statistical correlations between the system variables, where the matrix $P_0=P(0)$ pertains to the one-point covariances and satisfies the ALE 
  $AP_0 + P_0A^\rT + BB^\rT = 0$, 
similar to (\ref{PR}). The kernel (\ref{P}) gives rise to a positive semi-definite self-adjoint operator $\cP$ on $L^2([0,T], \mC^n)$ which maps a function $f$ to $g:= \cP(f)$ as
\begin{equation}
\label{cP}
    g(s):= \int_0^T P(s-t)f(t)\rd t,
    \qquad
    0\< s \< T
\end{equation}
(with $\cP+i\cL$ also being a positive semi-definite self-adjoint operator on $L^2([0,T], \mC^n)$).
In the case of the multipoint Gaussian quantum state, the MGF (\ref{Psi}) admits a closed-form representation \cite{VPJ_2018a}, and (\ref{XiZ}) takes the form
\begin{equation}
\label{XiZ1}
  \Xi
    =
    \re^{-C}
    \bM
  \re^{\frac{\theta}{2}\int_{[0,T]^2} \rd Z(s)^\rT P(s-t)\rd Z(t)}.
\end{equation} 
The right-hand side of (\ref{XiZ1}) does not involve quantum variables and is organised as a quadratic-exponential moment for the auxiliary classical Gaussian random process $Z$ in (\ref{Z}), which can be computed as follows. 
Let $N$ be another independent auxiliary classical Gaussian random process on the time interval $[0,T]$ with values in $\mR^n$, zero mean and the covariance function (\ref{P}), so that $\bM (N(s)N(t)^\rT) = P(s-t)$ for all $0\< s,t\< T$. Then by using the MGFs for Gaussian random processes and the tower property of conditional classical expectations \cite{S_1996}, it follows that
\begin{align}
\nonumber
    \bM
  \re^{\frac{\theta}{2}\int_{[0,T]^2} \rd Z(s)^\rT P(s-t)\rd Z(t)}
   &=
  \bM
    \bM
    \big(
  \re^{\sqrt{\theta}\int_{[0,T]^2} N(t)^\rT\rd Z(t)}
  \mid
  Z\big)\\
\nonumber
  & =
  \bM
  \re^{\sqrt{\theta}\int_{[0,T]^2} N(t)^\rT\rd Z(t)}
  =
  \bM
    \bM
    \big(
  \re^{\sqrt{\theta}\int_{[0,T]^2} N(t)^\rT\rd Z(t)}
  \mid
  N\big)\\
\label{MMM}
    & =
    \bM
  \re^{\frac{\theta}{2}\bra N, \cK(N)\ket }
  =
  \re^{-\frac{1}{2} \Tr \ln (\cI - \theta \cP \cK)}
  =
  \prod_{k=1}^{+\infty}
  \tfrac{1}{\sqrt{1-\theta \lambda_k}},
\end{align}
where use is also made of  the covariance operator $\cK$  of the incremented process $Z$ in (\ref{cK}), (\ref{iso}) and the Fredholm determinant formula
\cite[Theorem~3.10 on p.~36]{S_2005} (see also \cite{G_1994} and references therein).
Here,  $\lambda_k\>0$ are the eigenvalues of $\cP\cK$ which is a compact operator on $L^2([0,T], \mR^n)$ (isospectral to the positive semi-definite self-adjoint operator $\sqrt{\cK} \cP \sqrt{\cK}$) as the composition of a bounded operator ($\cK$) and a compact operator ($\cP$ in (\ref{cP})). For the validity of (\ref{MMM}), the spectral radius condition
\begin{equation}
\label{spec}
    \br(\cP\cK) = \max_{k\in \mN} \lambda_k < \tfrac{1}{\theta}
\end{equation}
has to be satisfied.
Substitution of (\ref{C}), (\ref{cK}), (\ref{MMM}) into (\ref{XiZ1}) represents the QEF (\ref{Xi}) as
\begin{equation}
\label{XiZ2}
  \Xi
    =
    \re^{
    -  \frac{1}{2}
  \Tr (\ln\cos (\theta\cL) + \ln (\cI - \theta \cP\tanc (\theta\cL)))
  }.
\end{equation}
In the limiting classical case mentioned above, when $\cL = 0$ and $\cK=\cI$, the representation (\ref{XiZ2}) takes the form  $ \Xi=   \re^{
    -  \frac{1}{2}
  \Tr \ln (\cI - \theta \cP)
  }
$ and the condition (\ref{spec}) reduces to $\theta \br(\cP)<1$.

\section{\bf Conclusion}
\label{sec:conc}

For a multimode open quantum harmonic oscillator, we have considered a finite-horizon expansion of its system variables over the eigenfunctions of their two-point commutator kernel. This quantum Karhunen-Loeve expansion involves  pairs of noncommuting position and momentum operators (with interpair commutativity) as coefficients.  The orthogonality of the real and imaginary parts of the eigenfunctions has allowed for a single (rather than double) series representation of the integral-of-quadratic functions of the system variables. This has been used in order to obtain a Girsanov type representation for the quadratic-exponential functions of the system variables whose mean values provide robust performance criteria for quantum risk-sensitive control. This representation is valid regardless of a particular system-field state and employs an auxiliary  Gaussian random process and classical averaging over it. This result has been used in order to express the QEF in terms of the moment-generating functional of the system variables. In application to the invariant multipoint Gaussian quantum state of the oscillator driven by vacuum input fields, the QEF has been reduced to a quadratic-exponential moment of the auxiliary Gaussian random process. The frequency-domain and state-space computation of the QEF and its asymptotic growth rate will be considered in subsequent publications on this topic.

\appendix

\section*{\bf Randomised representation for quadratic-exponential functions of position-momentum pairs}
\renewcommand{\theequation}{A\arabic{equation}}
\setcounter{equation}{0}

For the purposes of Section~\ref{sec:quadro}, consider a self-adjoint quantum variable
\begin{equation}
\label{f}
  f(\sigma)
  :=
  \bM \re^{\sigma(\alpha \xi + \beta \eta )}
  =
  \tfrac{1}{2\pi}
  \int_{\mR^2}
    \re^{\sigma(a \xi + b\eta )
    -\frac{1}{2}(a^2 +b^2)}
  \rd a \rd b,
\end{equation}
which is associated with the position and momentum operators $\xi$, $\eta$ (see Section~\ref{sec:sys}) and depends on a real-valued parameter $\sigma$ satisfying
\begin{equation}
\label{sigmagood}
    |\sigma| < \sqrt{2}.
\end{equation}
Here,
the classical expectation $\bM(\cdot)$ is over auxiliary independent standard normal  random variables $\alpha$, $\beta$. The following theorem relates the quantum variables (\ref{f}) with ``elementary'' quadratic-exponential functions  of the position-momentum pair $(\xi, \eta)$.

\begin{thm}
\label{th:fact}
The quantum variables (\ref{f}) satisfy the identity
\begin{equation}
  \label{qefrand}
    \re^{\omega (\xi^2+\eta^2)}
    =
    \tfrac{1}{\cosh \omega}
    f(\sigma),
    \qquad
    \sigma := \sqrt{2\tanh \omega}
  \end{equation}
  for any $\omega\>0$.
\hfill$\square$
\end{thm}
\begin{proof}
Differentiation of (\ref{f}) with respect to $\sigma$  yields
\begin{equation}
\label{f'}
    f'
     =
    \bM ((\alpha \xi + \beta \eta ) g)\\
    =
    \zeta^\rT
    \mu,
\end{equation}
where
\begin{equation}
\label{BCH}
    g
    :=
    \re^{\sigma(\alpha \xi + \beta \eta )}
    =
    \re^{\sigma\alpha (\xi-\frac{i}{2}\sigma \beta)}
    \re^{\sigma\beta \eta }
    =
    \re^{\sigma\beta (\eta+\frac{i}{2}\sigma \alpha)}
    \re^{\sigma\alpha\xi },
\end{equation}
and the position-momentum vector (\ref{zeta}) is used together with
\begin{equation}
\label{mu}
    \mu
    :=
      {\begin{bmatrix}
      \mu_1\\
      \mu_2
    \end{bmatrix}}
    =
     {\begin{bmatrix}
      \bM(\alpha g)\\
      \bM(\beta g)
    \end{bmatrix}}.
\end{equation}
The second and third equalities in (\ref{BCH}) follow from the Baker-Campbell-Hausdorff (BCH) formula (or, equivalently, the Weyl CCRs, similar to (\ref{Weyl})). Since $\alpha$, $\beta$ are independent standard normal random variables, application of the tower property of conditional classical expectations \cite{S_1996} leads to
\begin{align}
\nonumber
    \mu_1
    & =
    \bM
    (
    \alpha
    \re^{\sigma\alpha (\xi-\frac{i}{2}\sigma \beta)}
    \re^{\sigma\beta \eta }
    )\\
\nonumber
    &=
    \bM
    (
    \bM
    (
    \alpha
    \re^{\sigma\alpha (\xi-\frac{i}{2}\sigma \beta)}
    \mid
    \beta
    )
    \re^{\sigma\beta \eta }
    )\\
\nonumber
    &=
    \sigma
    \bM
    (
    \bM
    (
    (\xi-\tfrac{i}{2}\sigma \beta)
    \re^{\sigma\alpha (\xi-\frac{i}{2}\sigma \beta)}
    \mid
    \beta
    )
    \re^{\sigma\beta \eta }
    )    \\
\nonumber
    &=
    \sigma
    \bM
    (
    (\xi-\tfrac{i}{2}\sigma \beta)
    \re^{\sigma\alpha (\xi-\frac{i}{2}\sigma \beta)}
    \re^{\sigma\beta \eta }
    )    \\
\label{Mg1}
    &=
    \sigma
    \bM
    (
    (\xi-\tfrac{i}{2}\sigma \beta)
    g
    )
    =
    \sigma
    \xi
    f
    -
    \tfrac{i}{2}
    \sigma^2
    \mu_2,
\end{align}
where use is made of the second equality from (\ref{BCH}) together with the identity
\begin{equation}
\label{nice}
    \bM (\gamma \re^{\gamma z}) = z\re^{\frac{1}{2}z^2} = z \bM \re^{\gamma z},
\end{equation}
which holds for a standard normal random variable $\gamma$ and extends from complex numbers to operators $z$. In a similar fashion, a combination of the third equality from (\ref{BCH}) with (\ref{nice}) yields
\begin{align}
\nonumber
    \mu_2
    & =
    \bM
    (
    \beta
    \re^{\sigma\beta (\eta+\frac{i}{2}\sigma \alpha)}
    \re^{\sigma\alpha\xi }
    )\\
\nonumber
    &=
    \bM
    (
    \bM
    (
    \beta
    \re^{\sigma\beta (\eta+\frac{i}{2}\sigma \alpha)}
    \mid
    \alpha
    )
    \re^{\sigma\alpha\xi }
    )\\
\nonumber
    &=
    \sigma
    \bM
    (
    \bM
    (
    (\eta+\tfrac{i}{2}\sigma \alpha)
    \re^{\sigma\beta (\eta+\frac{i}{2}\sigma \alpha)}
    \mid
    \alpha
    )
    \re^{\sigma\alpha\xi }
    )    \\
\nonumber
    &=
    \sigma
    \bM
    (
    (\eta+\tfrac{i}{2}\sigma \alpha)
    \re^{\sigma\beta (\eta+\frac{i}{2}\sigma \alpha)}
    \re^{\sigma\alpha\xi }
    )    \\
\label{Mg2}
    &=
    \sigma
    \bM
    (
    (\eta+\tfrac{i}{2}\sigma \alpha)
    g
    )
    =
    \sigma
    \eta
    f
    +
    \tfrac{i}{2}
    \sigma^2
    \mu_1.
\end{align}
The relations (\ref{Mg1}), (\ref{Mg2}) form a set of two linear equations for the vector $\mu$ in  (\ref{mu}), which can be represented by using the matrix (\ref{bJ}) and the vector (\ref{zeta}) as
$
    \mu
    =
    \sigma
    \zeta
    f
    -
    \tfrac{i}{2}
    \sigma^2
    \bJ
    \mu
$,
and hence,
\begin{equation}
\label{mu1}
    \mu
    =
    \sigma
    \big(
    I_2
    +
    \tfrac{i}{2}
    \sigma^2
    \bJ
    \big)^{-1}
    \zeta
    f
    =
    \tfrac{\sigma}{1-\frac{1}{4} \sigma^4}
    \big(I_2
    -
    \tfrac{i}{2}
    \sigma^2
    \bJ
    \big)
    \zeta
    f
\end{equation}
in view of the idempotence $(i\bJ)^2 = I_2$. Substitution of (\ref{mu1}) into (\ref{f'}) leads to
\begin{equation}
\label{f'1}
    f'
    =
    \tfrac{\sigma}{1-\frac{1}{4} \sigma^4}
    \zeta^\rT
    \big(I_2
    -
    \tfrac{i}{2}
    \sigma^2
    \bJ
    \big)
    \zeta
    f
    =
    \tfrac{\sigma}{1-\frac{1}{4} \sigma^4}
    \big(
    \zeta^\rT \zeta
    +
    \tfrac{1}{2}
    \sigma^2
    \big)
    f,
\end{equation}
where the CCR $\zeta^{\rT}\bJ \zeta = [\xi, \eta] = i$ is used. Here, the denominator $1-\frac{1}{4} \sigma^4$ originates from the BCH correction factors $\re^{\pm \frac{i}{2}\sigma^2 \alpha \beta}$ in (\ref{BCH}), which explains the nature of the constraint (\ref{sigmagood}).
Now, (\ref{f'1}) is an operator differential equation with the identity operator as the initial condition $f(0)=\cI$. It can be integrated by using a new variable $\omega\>0$, related to $\sigma$  as
\begin{equation}
\label{lamsig}
    \omega := \tfrac{1}{2}\ln\tfrac{1+\frac{1}{2}\sigma^2}{1-\frac{1}{2}\sigma^2},
    \qquad
    \sigma = \sqrt{2\tfrac{\re^{2\omega}-1}{\re^{2\omega}+1}} = \sqrt{2\tanh \omega}.
\end{equation}
Indeed, $\tfrac{\sigma \rd \sigma}{1-\frac{1}{4} \sigma^4} = \tfrac{\rd (\sigma^2/2)}{1-\frac{1}{4} \sigma^4} = \rd \omega$, and hence, (\ref{f'1}) can be represented in the form $\rd f = f'\rd \sigma = (\zeta^\rT \zeta + \tanh \omega )f\rd \omega$ whose solution is given by
\begin{equation}
\label{qeff}
    f(\sigma) = \re^{\int_0^\omega\tanh u \rd u }\re^{\omega \zeta^\rT \zeta} f(0)
    =
    \re^{\omega \zeta^\rT \zeta} \cosh \omega.
\end{equation}
With (\ref{lamsig}) describing a bijection $[0,+\infty)\ni \omega \leftrightarrow \sigma \in [0,\sqrt{2})$, it now remains to note that $\zeta^\rT \zeta = \xi^2 + \eta^2$, whereby (\ref{qeff}) establishes (\ref{qefrand}).
\end{proof}

\end{document}